\newcommand{\ecversion}[2]{#2}{} 

\ecversion{
\documentclass[prodmode,acmec]{ec-acmsmall}
\acmVolume{X}
\acmNumber{X}
\acmArticle{X}
\acmYear{2013}
\acmMonth{2}
\usepackage[numbers]{natbib}
\bibliographystyle{acmsmall}
}{

\documentclass[11pt,letterpaper]{article}
\topmargin=-0.35in \topskip=0pt \headsep=15pt \oddsidemargin=0pt
\textheight=9in \textwidth=6.52in \voffset=0in
\bibliographystyle{abbrv}

}

\usepackage{amsmath,amsfonts,graphpap,amscd,mathrsfs,graphicx,lscape}
\usepackage{epsfig,amssymb,amstext,xspace}
\ecversion{}{
\usepackage{theorem}
}
\usepackage{bbm}
\usepackage{algorithm,algorithmic}

\usepackage{color}              
\usepackage{epsfig}

\newcommand{\abs}[1]{\vert{#1}\vert}

\ecversion{}{
\newtheorem{theorem}{Theorem}[section]
\newtheorem{lemma}[theorem]{Lemma}

\newtheorem{corollary}[theorem]{Corollary}

{\theorembodyfont{\rmfamily} }

}

\newtheorem{defn}[theorem]{Definition}

\def\squareforqed{\hbox{\rule{2.5mm}{2.5mm}}}

\def\QED{\ifmmode\squareforqed 
  \else{\nobreak\hfil   
    \penalty50                 
    \hskip1em                  
    \null                      
    \nobreak                   
    \hfil                      
    \squareforqed              
    \parfillskip=0pt           
    \finalhyphendemerits=0     
    \endgraf}                  
  \fi}

\def\blksquare{\rule{2mm}{2mm}}
\def\qedsymbol{\blksquare}
\newcommand{\bg}[1]{\medskip\noindent{\bf #1}}
\newcommand{\ed}{{\hfill\qedsymbol}\medskip}

\ecversion{}{
\newenvironment{proof}{\bg{Proof : }}{\ed}
}
\newenvironment{proofof}[1]{\bg{Proof of #1 : }}{\ed}

\ecversion{}{
\newenvironment{example}{\bg{Example. }}{\ed}
}

\newcommand{\R}{\ensuremath{\mathbb R}}

\newcommand{\comment}[1]{}
{}  

\newcommand{\junk}[1]{}

\newlength{\tmp} \newlength{\lpsx} \newlength{\lpsy} \newlength{\upsx}
\newlength{\upsy}

\newcommand{\email}[1]{\texttt{#1}}

\newcommand{\sw}{\mathbf{W}}

\newcommand{\csw}{\mathbf{\bar{W}}}
\newcommand{\cl}{\mathtt{c}}
\newcommand{\up}{\mathtt{u}}
\newcommand{\bv}{\bar{v}}

\begin{document}

\setcounter{page}{0}

\title{Efficiency Guarantees in Auctions with Budgets}
\ecversion{
\author{SHAHAR DOBZINSKI
\affil{Weizmann Institute of Science}
RENATO PAES LEME
\affil{Microsoft Research SVC}
}
}{
\author{
Shahar Dobzinski\\
       Weizmann Institute\\
       \email{dobzin@gmail.com}
\and
Renato Paes Leme\\
       Microsoft Research SVC\\
       \email{renatoppl@cs.cornell.edu}
}
}

\date{}

\ecversion{}{\maketitle}

\begin{abstract}
In settings where players have a limited access to liquidity, represented in the
form of budget constraints, efficiency maximization has proven to be a
challenging goal. In particular, the social welfare cannot be approximated by a
better factor then the number of players. Therefore, the literature has mainly
resorted to Pareto-efficiency as a way to achieve efficiency in such settings.
While successful in some important scenarios, in many settings it is known that
either exactly one incentive-compatible auction that always outputs a
Pareto-efficient solution, or that no truthful mechanism can always guarantee a
Pareto-efficient outcome.
Traditionally, impossibility results can be avoided by considering
approximations. However, Pareto-efficiency is a binary property
(is either satisfied or not), which does not allow for approximations.

In this paper we propose a new notion of efficiency, called \emph{liquid
welfare}. This is the maximum amount of revenue an omniscient seller would be
able to extract  from a certain instance. We explain the intuition behind this
objective function and show that it can be $2$-approximated by two different
auctions. Moreover, we show that no truthful algorithm can guarantee an
approximation factor better than $4/3$ with respect to the liquid welfare, and
provide a truthful auction that attains this bound in a special case.

Importantly, the liquid welfare benchmark also overcomes impossibilities for some
settings. While it is impossible to design Pareto-efficient auctions
for multi-unit auctions where players have decreasing marginal values, we give
a deterministic $O(\log n)$-approximation for the liquid welfare in this
setting.
\end{abstract}

\ecversion{\maketitle}{}

\ecversion{}{
\renewcommand{\thepage}{}
\clearpage
\pagenumbering{arabic}

}

\section{Introduction}

\comment{
Budget constraints play a major role in many settings of interest: Benoit and
Krishna \cite{benoit_krishna} point out the importance of budgets in
privatization auctions in Eastern Europe and in the sale of spectrum in
the US, and more generally in settings where the
magnitude of values involved is so large that it might exhaust the players
liquid assets. Pai and Vohra \cite{pai_vohra} illustrate this situation by
saying that ``not every potential buyer of a David painting who values it at a
million dollars has access to a million dollars to make the bid.'' \comment{They
also make the case tha budget constraints are more concrete and palatable
than valuations.} Another practical setting for which budget constraints are of
foremost relevance is internet advertising. In fact, in the interface of Google
Adwords the choice of budget to spend is the first question asked to
advertisers, even before they
are asked bids or keywords. We refer to  Aggarwal et al \cite{Aggarwal09}, Goel et al
\cite{goel12} and Colini-Baldeschi et al \cite{henzinger11} for a
discussion of the role of the budgets in internet advertisement. For a more
extensive discussion on the source of financial constraints, we refer to Che
and Gale \cite{che_gale}.
}


{
Budget constraints play a major role in many settings of interest, being internet
advertisement one of the most important examples. In fact, the choice of budget to spend is the first question asked to
advertisers in the interface of
Google Adwords, even before they are asked bids or keywords. Much work has been
devoted to understanding the impact of budget constraints in sponsored search
auctions. See. e.g., \cite{Aggarwal09, goel12, henzinger11}.
Budgets are generally important in settings where the magnitude of the
financial transactions is very large. Benoit and Krishna
 \cite{benoit_krishna} point out the importance of budgets in
privatization auctions in Eastern Europe and in the sale of spectrum in
the US, and more generally in settings where the
magnitude of values involved is so large that it might exhaust the players
liquid assets. Pai and Vohra \cite{pai_vohra} illustrate this situation by
saying that ``not every potential buyer of a David painting who values it at a
million dollars has access to a million dollars to make the bid.'' \comment{They
also make the case that budget constraints are more concrete and palatable
than valuations.}  For a more extensive discussion on the source of financial
constraints, we refer to Che and Gale \cite{che_gale}.
}

Due to its practical relevance, it is not surprising that so many theoretical
investigations have been devoting to analyzing auctions for budget
constrained agents. For analyzing the impact of budget in the revenue of
standard auctions one can cite \cite{che_gale} and  \cite{benoit_krishna}, and
for designing mechanims optimize or
approximately optimize revenue, one can cite:
\cite{laffont_robert, malakhov_vohra, pai_vohra, Borgs05, chawla11}.

When the objective is welfare efficiency rather then revenue, the literature
has early stumbled upon impossibility results. The traditional social welfare
measure, the sum of player's values for their outcomes, is known not to work
well under budget constraints. Indeed, a folklore result shows that even if
budgets are known to the auctioneer, no incentive-compatible auction
can achieve a better than $n$ approximation for the social welfare, where $n$ is
the number of players. This motivates the search for incentive-compatible
auctions satisfying weaker notions of efficiency. Dobzinski, Lavi and Nisan
\cite{dobzinski12} suggest studying Pareto-efficient auctions: the outcome
of an auction is said to be Pareto-efficient if there it no alternative outcome
(allocation and payments) where no agent (bidders or auctioneer) are worse-off
and at least one agent is better off. If budgets are public, the authors give an
incentive-compatible and
Pareto-efficient multi-unit auction based on the Ausubel's clinching framework
\cite{Ausubel_multi}. Furthermore, they show that this auction is the unique
truthful auction that always produces Pareto-efficient solution. The study of
Pareto-efficient auctions for budget
constrained players has been extended to different settings in a sequence of
follow-up papers \cite{Bhattacharya10, fiat_clinching, henzinger11, goel12,
goel13}.

\subsection*{Beyond Pareto Efficiency?}

It therefore seems that Pareto efficiency has emerged as the de-facto standard
for measuring efficiency in settings where bidders are budget constrained.
Indeed, most of the aforementioned papers and results provide positive results
by offering new auctions. Yet, this is far from being a complete solution from
both theoretical and practical point of views. We now elaborate on this issue.

In a sense, the uniqueness result of Dobzinski et al \cite{dobzinski12} for
public budgets -- that shows that the clinching auction is
the only Pareto efficient, truthful auction -- may be viewed negatively. Rare are
the cases in practice in which the designer sole goal is to obtain a Pareto
efficient allocation. A more realistic view is that theory provides the
designer a toolbox of complementary methods and techniques designed to obtain
various different goals (efficiency, revenue maximization, fairness,
computational efficiency, etc.) and balance between them. The composition of
these tools as well as their adaptation to the specifics of the setting
and fine tuning is the designer's task. A uniqueness result -- although
extremely appealing from a pure theoretical perspective -- implies that the
designer's toolbox contains only one tool, obviously an undesirable scenario.

Furthermore, although from a technical point of view the analysis of the
existing algorithms is very challenging and the proof techniques developed are
quite unique to each setting, the auctions themselves are all variants of the
same basic clinching idea of Ausubel \cite{Ausubel_multi}. Again, it is
obviously preferable to have more than just one bunny in the hat that will help
us design auctions for these important settings.

The situation is obviously even more severe in more complicated settings, where
even this lonely bunny is
not available. For example, for private budgets and additive multi-unit auctions,
an impossibility was given by \cite{dobzinski12}, for heterogenous items and
public budgets by  \cite{fiat_clinching, dutting12}
and for
as multi-unit auctions with subadditive valuations and public budgets,
by Goel, Mirrokni and Paes Leme \cite{goel12, Lavi_May}.

\subsection*{Alternatives to Pareto Efficiency}

Our main goal is to research alternatives to Pareto efficiency for budget
constrained agents. We start from the observation that a Pareto efficient
solution is a binary notion: an allocation is either Pareto efficient or not,
and there is no sense of one allocation being ``more Pareto efficient'' than the
other. This is in contrast with efficiency in quasi linear environments
where the traditional welfare objective induces a total order on the
the allocations.

Our main goal in this paper is to provide a new measure of efficiency for
budgeted settings. The desiderata for this measure are: (i) it is quantifiable,
i.e., attaches a value to each outcome; (ii) is achievable, i.e., can be
approximated by incentive-compatible mechanisms and (iii) allows different
designs that approximate welfare.

The measure we propose is called the \emph{liquid welfare}. Before defining
it, we give a revenue-motivated definition of the traditional social welfare in
unbudgeted settings and show how it naturally generalizes to budgeted settings.
One can view the traditional welfare of a certain outcome as the maximum
revenue an omniscient seller can obtain from that outcome. If each agent $i$
has value $v_i(x_i)$ for a certain outcome $x_i$, the omniscient seller can
extract revenue arbitrarily close to $\sum_i v_i(x_i)$ by offering this outcome
to each player $i$ for price $v_i(x_i)-\epsilon$. This definition generalizes
naturally to budgeted settings. Given an outcome, $x_i$, the
\emph{willingness-to-pay} of agent $i$ is $v_i(x_i)$, which is the maximum he
would give for this outcome in case he had unlimited resources. His
\emph{ability-to-pay}, however, is $B_i$, which is the maximum amount of money
available to him. We define his \emph{admissibility-to-pay} as the maximum
value he would admit to pay for this outcome, which is the minimum between his
willingness-to-pay and his ability-to-pay. The liquid welfare of a certain
outcome is defined as the total admissibility-to-pay. Formally $\csw(x) =
\sum_i \min(v_i(x_i), B_i)$.

An alternative point of view is as follows: efficiency should be measured only
with respect to the funds available to the bidder at the time of the auction,
and not the additional liquidity he might gain after receiving the goods in the
auction. When using the liquid welfare objective, the auctioneer is therefore
freed from considering the hypothetical use the bidders will make of the items
they win in the auction and can focus only on the resources available to them
at the time of the auction.

This objective satisfies our first requirement: it associates each outcome with
an objective measure. Also, it is achievable. In fact, the clinching
auction \cite{dobzinski12}, which is the base for all auction achieving
Pareto-efficient outcomes for budgeted settings, provide a $2$-approximation
for the liquid welfare objective. To show that this allows flexibility in the
design, we show a different auction that also provides a $2$-approximation and
reveals a connection between our liquid welfare objective and the notion
of market equilibrium.

It is appropriate to discuss the applicability and limitations of the liquid
welfare objective. We would like to begin by illustrating a setting for which
it is \emph{not} applicable. If one were to auction hospital beds or access to
doctors, it would be morally repugnant to privilege players based on their
ability-to-pay. Therefore, we are not interested in claiming that the liquid
welfare objective is the only alternative to Pareto efficiency, but rather argue
that in \emph{some} settings it produces reasonable results. Developing other
notions of efficiency is an important future direction.

Still, in many settings capping the welfare of the agents by
their budgets makes perfect sense. Consider designing a market like internet
advertising which aims at a healthy mix of good efficiency and revenue. In
practice, players that bring more money to the market provide health to the
market and make it more efficient. In real markets, there are practices to
encourage wealthier players to enter the market. Therefore, privileging such
players in the objective is somewhat natural.

An interesting question is whether one can have a truthful mechanism for
additive valuations with public budgets that provides an approximation ratio
better than $2$ approximation. We show a lower bound of $\frac 4 3$. Closing the
gap remains an open question, but we do show that for the special case of $2$
players with identical public budgets there is a truthful auction that provides
a matching upper bound.

We then move one to consider a setting in which truthful auction that always
output Pareto-efficient solution do not exist: indivisible multi-unit auctions
with additive valuations and private budgets. For this setting we borrow ideas
from Bartal, Gonen and Nisan \cite{BGN} and provide a deterministic
$O(\log^2 n)$ approximation to the liquid welfare. In fact, the algorithm is
even more powerful since the approximation ratio holds even if
the valuations of the bidders are known to be subadditive.

\subsection*{Related Work}
We have already surveyed results designing mechanisms for
budget-constrained agents. Here, we focus on surveying results directly related
to our efficiency measure and to our philosophical approach to efficiency
maximization. As far as we know, the liquid welfare was first appeared in
Chawla, Malec and Malekian
\cite{chawla11} for the first time as an implicit upper bound on the revenue
that a mechanism can extract.

Independently and simultaneously, two other approaches were proposed to provide
quantitative guarantees for budgeted settings. Devanur, Ha and Hartline
\cite{devanur12} show that the welfare of the clinching auction is a
$2$-approximation to the welfare of the best envy-free equilibrium. Their
approach, however, is restricted to settings with \emph{common budgets}, i.e.,
all agents have the same budget.

Syrgkanis and Tardos \cite{syrgkanis12} leave the realm of truthful
mechanisms
and study the set of Nash and Bayes-Nash equilibria of simple mechanisms. For a
wide class of mechanisms they show that the \emph{traditional} welfare in
equilibrium of such mechanism is a constant fraction of the optimal liquid
welfare objective (which they call \emph{effective welfare}). {
Their approach differs from ours in two ways: first they study auctions
in equilibrium while we focus on incentive compatible auctions. Second,
the guarantee in their mechanism is that the welfare of the allocation obtained
is always greater than some fraction of the liquid welfare.
The guarantee of our
mechanisms is stronger: we construct mechanisms in which the \emph{liquid welfare}
is always greater than some fraction of the liquid welfare, which implies in
particular that the welfare is greater than some fraction of the liquid welfare
(since the welfare of an allocation is at least its liquid welfare).
}

\subsection*{Summary of Our Results}

In this paper we have proposed to study the liquid welfare. We have provided
two truthful algorithms that provide a $2$ approximation to this objective
function for the setting of multi-unit auctions with public budgets. For the
harder setting of multi-unit auctions with subadditive valuations and private
budgets we have provided a truthful algorithm that provides an $O(\log^2 n)$.
For submodular bidders, the same mechanism provides an $O(\log n)$
approximation.

The biggest immediate problem that we leave open is to determine whether there
is an algorithm for multi unit auction with private budgets that obtain a
constant approximation ratio. The question is even open if the valuations are
known to be additive. More generally, are there truthful algorithms that
provides a good approximation for \emph{combinatorial} auctions? On top of that,
notice that computational considerations might come into play: while all of the
constructions that we present in this paper happen to be computationally
efficient, there might be a gap between the power of truthful algorithms in
general and the power of computationally efficient truthful algorithms.

\comment{
The main message of this paper is that there might be reasonable ways to
quantify the quality of allocations in setting where bidders have budget
constraints. We studied the liquid welfare, but for some settings other
notions might be more appropriate. We hope that this paper will trigger the
study of other objective functions, and their advantages and disadvantages.
}

\section{Preliminaries}\label{sec:preliminaries}

\subsection{Environments of interest and auction basics}
We consider $n$ players and a set $X$ of outcomes (also called environment). For
each player, let $v_i : X \rightarrow \R_+$ be the valuation function for player
$i$. We consider that agents are budgeted quasi-linear, i.e., each agent $i$ has
a budget $B_i$ and for an outcome $x$ and for payments $\pi_1, \hdots, \pi_n$,
the utility of agent $i$ is: $u_i = v_i(x) - \pi_i $ if $\pi_i \leq B_i$ and
$-\infty$ o.w. Below, we list a set of environments we are interested:

\begin{enumerate}
\item \emph{Divisible-multi-unit auctions and additive bidders:} $X
= \{(x_1, \hdots, x_n); \sum_i x_i = s\}$ for some constant $s$ and $v_i(x) =
v_i \cdot x_i$, so we can represent the valuation function of each agent by a
single real number $v_i \geq 0$.
\item \emph{Divisible-multi-unit auctions with subadditive bidders:} $X
= \{(x_1, \hdots, x_n); \sum_i x_i = s\}$ for some constant $s$ and $v_i(x) =
v_i(x_i)$, where $v_i : \R_+ \rightarrow \R_+$ is a monotone non-decreasing
subadditive function. A valuation function $v_i$ is subadditive
if for every $x_1,x_2$, $v_i(x_1 + x_2) \leq v_i(x_1)+v_i(x_2)$.
 \item \emph{$0/1$ environments:} $X \subseteq \{0,1\}^n$ and $v_i(x) = v_i$
if $x_i = 1$ and $v_i(x) = 0$ otherwise. Again the valuation is represented by a
single $v_i \geq 0$.
\end{enumerate}

An auction for a particular setting elicits the valuations of the players and
budgets $B_1, \hdots, B_n$ and outputs an outcome $x \in X$ and payments
$\pi_1, \hdots, \pi_n$ for each agents respecting budgets, i.e., such that
$\pi_i \leq B_i$ for each agent. We will distinguish between \emph{public
budgets} and \emph{private budgets} mechanisms. In the former, the auctioneer
has access to the true budget of each agent\footnote{Most of the papers in
the literature on efficient auctions for budgeted settings \cite{dobzinski12,
fiat_clinching, goel12, henzinger11, goel13} fall in this category, including
classical references as, Laffont and Robert \cite{laffont_robert}
and Maskin \cite{maskin00}.}. In the later case,
agents need to be incentivized to report their true budget. In either case, the
valuations of each agent are private. We will focus on designing mechanisms
that are incentive compatible (a.k.a. truthful), i.e., are such that agents
utilities are maximized once they report their true value  in the public budget
case and their true value and budget in the private budget case. We will also
require mechanisms to be individually rational, i.e., agents always derive
non-negative utility upon bidding their true value.

In the case of divisible multi-unit auctions and additive bidders, the
valuations can be represented by real numbers, So we can see the auctions as a
pair of functions $x:\R_+^n \times \R_+^n \rightarrow \R_+^n$ and $\pi:\R_+^n
\times \R_+^n \rightarrow \R_+^n$ that map $(v,B)$ to a vector of allocations
$x(v,B) \in \R^n_+$ and a vector of payments $\pi(v,B) \in \R^n_+$. The set of
functions that induce incentive compatible and individually rational auctions
are characterized by Myerson's Lemma:

\begin{lemma}[Myerson \cite{myerson-81}]\label{lemma:myerson}
 A pair of functions $(x,\pi)$ define an incentive-compatible and individually
rational auction iff (i) for each $v_{-i}$, $x_i(v_i,v_{-i})$ is monotone
non-decreasing in $v_i$ and (ii) the payments are such that: $\pi_i(v_i,
v_{-i}) = v_i \cdot x_i(v_i, v_{-i}) - \int_0^{v_i} x_i(u, v_{-i}) du$.
\end{lemma}

\subsection{Efficiency measures}
The traditional efficiency measure in mechanism design is the \emph{social
welfare} which associates for each outcome $x$, the objective: $\sw(x) =
\textstyle\sum_i v_i(x)$. It is known that one cannot even approximate the
optimal welfare in budgeted settings in an incentive-compatible way, even if
the budgets are known and equal. The result is folklore and we only
sketch the proof here for completeness:

\begin{lemma}[folklore]\label{lemma:welfare_innaprox}
Consider the divisible-multi-unit auctions and additive bidders. There is no
$\alpha$-approximate, incentive compatible and individually rational mechanism
$x(v), \pi(v)$ with $\alpha < {n}$. For $\alpha = {n}$ there is
the mechanism that allocates the item at random to one player and charges
nothing.
\end{lemma}

\begin{proof}
Let all the agents have budget $B_i = 1$. First notice that $\lim_{v_i
\rightarrow \infty} x_i(v_i, v_{-i}) \geq \alpha^{-1}$ since $\sum_j v_j \cdot
x_j(v) \geq \alpha^{-1} \max_i v_i \geq \alpha^{-1} v_i$, so: $x_i(v) + \sum_{j \neq i}
\frac{v_j}{v_i} \cdot x_j(v) \geq \alpha^{-1}$. Taking $v_i \rightarrow \infty$ we
get that $\lim_{v_i
\rightarrow \infty} x_i(v_i, v_{-i}) \geq \alpha^{-1}$.

By incentive compatibility, $v_i x_i(v) - \pi_i(v) \geq v_i x_i(v'_i, v_{-i})
- \pi_i(v'_i, v_{-i})$. Using the fact that $0 \leq \pi_i \leq 1$ we have: $v_i
x_i(v) \geq v_i x_i(v'_i, v_{-i}) - 1$ so: $x_i(v) \geq x_i(v'_i, v_{-i})
- \frac{1}{v_i}$. Taking $v'_i \rightarrow \infty$ we get: $x_i(v) \geq \alpha^{-1}
- \frac{1}{v_i}$. Summing for all players we get $1 \geq \sum_i x_i(v) \geq n
\alpha^{-1} - \sum_i \frac{1}{v_i}$. Taking $v_i \rightarrow \infty$ for all $i$, we
get $n \alpha^{-1} \leq 1$.
\end{proof}

Due to impossibility results of this flavor, efficiency was mainly achieved in
the literature through \emph{Pareto efficiency}. We say that an outcome
$(x,\pi)$ with $x \in X$ and $\pi_i \leq B_i$ is Pareto-efficient if there is no
alternative outcome where the utility of all the agents involved (including the
auctioneer, being his utility the revenue $\sum_i \pi_i$) does not decrease and
at least one agent improves. Formally, $(x,\pi)$ is Pareto optimal iff there is
no $(x', \pi')$, $x' \in X$, $\pi'_i \leq B_i$ such that:
$$u'_i = v_i \cdot x'_i - \pi'_i \geq u_i = v_i \cdot x_i - \pi_i, \forall i
\quad \text{ and } \quad \textstyle\sum_i \pi'_i \geq \textstyle\sum_i \pi_i
\quad \text{ and } \quad \sum_i v_i x'_i > \sum_i v_i x_i $$
In particular, if the budgets are infinity (or simply very large), then the only
Pareto-optimal outcomes are those maximizing social welfare.  For the setting of
divisible-multi-unit auctions with additive bidders, this is achieved by the
Adaptive Clinching Auction of Dobzinski, Lavi and Nisan \cite{dobzinski12}.
Moreover, the authors show that this is the only incentive-compatible,
individually-rational auction that achieves Pareto-optimal outcomes. The
auction is further analyzed in Bhattacharya et al \cite{Bhattacharya10} and Goel
et al \cite{goel13}.

In this paper we propose the \emph{liquid welfare} objective function:

\begin{defn}[Liquid Welfare]
 In a budgeted setting, we define the liquid welfare associated with outcome
$x
\in X$ by $\csw(x) = \sum_i \min\{ v_i(x), B_i \}$.
\end{defn}

We will refer to the optimal liquid welfare as $\csw^* = \max_{x \in X}
\csw(x)$. It is instructive yet straightforward to see that:

\begin{lemma}\label{lemma:x_star_format}
For divisible-multi-unit auctions and additive bidders, the optimal liquid
welfare $\csw^*$ occurs for $\bar{x}^*_i = \min \left( \frac{B_i}{v_i},
[1-\sum_{j < i} \bar{x}_j^*]^+ \right)$ where players are sorted in
non-increasing
order of value, i.e., $v_1 \geq v_2 \geq \hdots \geq v_n$.
\end{lemma}

An easy observation is the optimal allocation
for $\csw^*$ is not monotone in $v_i$, and hence cannot be implemented by a truthful auction. For example, consider $3$ agents with
values $v_1 = v, v_2 = 1, v_3 = 2$ and budgets $B_1 = 1, B_2 = \frac{1}{4}, B_3
= 1$. Now, it is simple to see that $\bar{x}^*_1(v_1)$ is not monotone in $v_1$
as depicted in Figure \ref{fig:xstar}.

\begin{figure}
\centering
\includegraphics{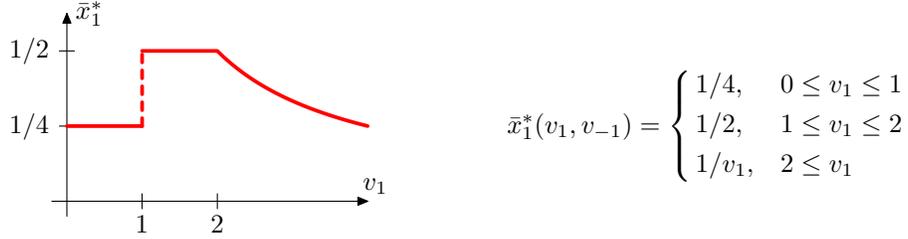}
\caption{Depiction of the first component of $\bar{x}^* = \text{argmax}_x
\csw(x)$
for a $3$ agent instance with $v = (v_1, 1,2)$ and $B = (1,1/4,1)$. The figure
highlights the non-monotonicity of the optimal solution $\bar{x}^*(v)$. }
\label{fig:xstar}
\end{figure}

\subsection{VCG and the liquid welfare objective}

The reader might suspect, however, that a modification of VCG
might take care of optimizing the liquid welfare benchmark. This is indeed
true for a couple of very simple settings. For example, for selling one
indivisible item, a simple Vickrey auction on modified values: $\bar{v}_i =
\min\{v_i, B_i\}$ provides an incentive compatible mechanism that exactly
optimized the liquid welfare objective. More generally:

\begin{theorem}[0/1 environments]
 Given a $0/1$-environment $X \subseteq \{0,1\}^n$ with valuations $v_i(x) =
v_i$ if $x_i = 1$ and zero otherwise. Then running VCG on modified values
$\bar{v}_i =\min\{v_i, B_i\}$ is incentive compatible and exactly optimized the
liquid welfare objective $\csw^*$.
\end{theorem}

The proof is trivial. This slightly generalizes to other simple environments of
interest, for example, matching markets, where there are $n$ agents and $n$
indivible items and each agent $i$ has a value $v_{ij}$ for item $j$ and
possible outcomes are perfect matchings. Running VCG on $\bar{v}_{ij} = \min
\{v_{ij}, B_i \}$ provides an incentive compatible mechanism that also exactly
approximated $\csw^*$.

This technique, however, does not generalize past those few special cases:

\begin{example}
Consider the problem of allocating $1$ divisible good among $n$ agents. One
possible VCG-reduction is to run VCG assuming that the agents values are
$\bar{v}_i(x_i) = \min \{B_i, v_i \cdot x_i\}$. This approach breaks since the
resulting mechanism is not monotone (Figure \ref{fig:xstar}) and hence not
incentive compatible. Other VCG approach is to assume agents are quasi-linear
with values $\bar{v}_i = \min\{v_i, B_i\}$. This  approach generates an
incentive compatible mechanism, but fails to produce non-trivial approximation
guarantees. Consider for example on agent with $B_1 = 2$ and $v_1 = 2$ and
other $n-1$ agents with $v_i = 1000$ and $B_i = 1$. The reduction gives an
auction between quasi-linear agents with $\bar{v}_1 = 2$ and $\bar{v}_i = 1$
for $i=2..n$. The auction, therefore, allocates the entire item to player $1$,
generating liquid welfare $\csw(x) = 2$, while $\csw^* = O(n)$. The same
example works for allocating a large number of identical indivisible goods.
\end{example}

\section{A First $2$-Approximation: The Clinching Auction}\label{sec:clinching}

In the previous section we defined our proposal for an efficiency measure in
budgeted settings: the liquid welfare objective $\csw$. The second item in
the desiderata for a new efficiency measure is that it is achievable, i.e., it
could be optimized or well-approximated by an incentive compatible mechanism.
In this section we show a mechanism that provides a $2$-approximation for
the liquid welfare, while still producing Pareto-efficient outcomes. The
mechanism we use is the Adaptive Clinching Auction \cite{dobzinski12,
Bhattacharya10, goel13}. In the next section we will provide another truthful
auction that provides a $2$-approximation, and show that with respect to the
liquid welfare the new auction is actually better on an instance-by-instance
basis.

We begin by reviewing the auction. The Adaptive Clinching Auction is described
by means of an ascending price procedure: we consider a price clock, whose
price $p$ starts at zero and gradually ascends. Let's say that the clock
ascends in increments of $\epsilon$. Initially, the good is un-allocated and
each player has their entire initial budget. For each price, we calculate how
much each player would like to acquire at each given price, which is his
remaining budget over the price $p$ if the price is below his marginal value.
We call such value the \emph{demand}. For each player $i$ we let player $i$
acquire at price $p$ an amount $\delta_i$ of the good corresponding to the
difference between the un-allocated amount and the sum of demands of the other
players. Then $\delta_i$ is allocated to player $i$ and his budget is
subtracted by $p\cdot \delta_i$. In the limit as $\epsilon$ goes to zero, the
auction is decribed by means of the following ascending price procedure:

\begin{defn}[Adaptive Clinching Auction]\label{defn:clinching_auction}
Consider one unit of a divisible good to be sold to $n$ agents with valuations
$v_1, \hdots, v_n$ per
unit and budgets $B_1, \hdots, B_n$. We assume $v_i \neq v_j$ for each $i \neq
j$ \footnote{If two values are equal, perturb the values to $v_i - \epsilon^i$
and take the limit as  $\epsilon \rightarrow 0$. We refer to the appendix of
Goel, Mirrokni and Paes Leme \cite{goel13} for a detailed description of the
clinching auction when two values are the same without resorting to
perturbations of the initial input.}. The final outcome is calculated by means
of an ascending price procedure which can be cast as a differential equation:
define functions $x_i(p), \pi(p)$ defined for each $p \in \R_+$. Also, define
the auxiliary function $S(p) = 1 - \sum_i x_i(p)$ which represents the
\emph{remnant supply} and define the sets $A(p)$ and $C(p)$ which we call the
\emph{active set} and the \emph{clinching set} respectively:
$$A(p) = \{i; p < v_i\} \qquad C(p) = \{i \in A(p); S(p) = \textstyle\sum_{j \in
A(p) \setminus i} \textstyle\frac{B_j - \pi_j(p)}{p} \}$$
Now, $x_i(p)$ and $\pi_i(p)$ are defined by means of the following rules:
\begin{itemize}
 \item for $p \notin \{v_1,\hdots, v_n\}$, $x_i(p), \pi_i(p)$ are
differentiable and the derivatives for $i \in C(p)$ are given by $\partial_p
x_i(p) = S(p)/p$ for $\partial_p \pi_i(p) = S(p)$ and for $i \notin C(p)$
$\partial_p x_i(p) = \partial_p \pi_i(p) = 0$.
 \item given a function $f$ and a price $p$, we define by $f(p-) = \lim_{x
\uparrow p} f(x)$ and  $f(p+) = \lim_{x
\downarrow p} f(x)$. So, for $p = v_i$, we define: $\delta^i_j = \left[
S(v_i-) - \sum_{k \in A(v_i) \setminus j} \frac{B_k - \pi_j(v_i-)}{v_i}
\right]^+$ and $x_j(v_i +) = x_j(v_i) = x_j(v_i-) + \delta_j^i$, $\pi_j(v_i+) =
\pi_j(v_i) = \pi_j(v_i -) + v_i \cdot \delta_j^i$ for $j \in A(v_i)$ and
$x_j(v_i +) = x_j(v_i) = x_j(v_i-)$, $\pi_j(v_i+) =
\pi_j(v_i) = \pi_j(v_i -)$ for $j \notin A(v_i)$.
\end{itemize}

The final allocation and payments of the clinching auction are given by
$\lim_{p \rightarrow \infty} x_i(p)$ and $\lim_{p \rightarrow \infty}
\pi_i(p)$ respectively. Since the values of $x_i(p)$ and $\pi_i(p)$ are
constant for $p > \max_i v_i$, the limit is reached for finite price.
\end{defn}

Bhattacharya et al \cite{Bhattacharya10} show that the procedure described
above defines an incentive compatible, individually rational auction and that
the outcomes are always Pareto-efficient. Before proving that the clinching
auction is a good approximation to $\csw^*$ we define the concept of
\emph{clinching interval}, which will be central in our proof. The clinching
interval corresponds to to the prices for which the clinching auction allocates
goods. Formally:

\begin{defn}[Clinching Interval]\label{defn:clinching_interval}
Given agents with valuations $v_i$ and budgets $B_i$, let
$S(p)$ be the remnant
supply function in Definition \ref{defn:clinching_auction}. Let $p_0 = \inf \{p;
S(p) < 1\}$ and $p_f = \sup \{p; S(p)
> 0\}$. We call the interval $[p_0, p_f]$, the clinching interval.
\end{defn}

We will also use the following lemma from Bhattacharya et al
\cite{Bhattacharya10} (Lemmas 3.3, 3.4 and 3.5 in their paper) :

\begin{lemma}\label{lemma:clinching_properties}
If the second highest value is below the budget of the highest value agent,
the entire good is allocated to the highest value agent. Otherwise, there
is some agent $k$ for which the final price $p_f = v_k$. Moreover the following
facts
hold:
\begin{itemize}
 \item  $\forall j, x_j > 0 \Rightarrow v_j \geq v_k$.
 \item  $\forall j,v_j > v_k \Rightarrow \pi_j = B_j$.
 \item  $\forall i,j, v_i < v_k < v_j \Rightarrow \delta_j^i = B_i / v_i$ and
$\delta_j^k =
(B_k -
\pi_k) / v_k$.
\end{itemize}
\end{lemma}

Now we are ready to state and prove our main Theorem:

\begin{theorem}
The clinching auction is a $2$-approximation to the liquid welfare
objective. In other words, given $n$ agents with values per unit $v_i$ and
budgets $B_i$, let $x, \pi$ be the outcome of the clinching auction for such
input. Then, $\csw(x) \geq \frac{1}{2} \csw^*$.
\end{theorem}

\begin{proof}
We assume for simplicity that the values of the players are all distinct. The
case where two players have the same value can be handled by
perturbing the input and considering the limit on the size of the
perturbation. We assume agents are sorted by value, i.e., $v_1 > v_2 > \hdots >
v_n$.

In the case where $v_2 \leq B_1$, the clinching
auction produces the allocation $x_1 = 1$
according to Lemma \ref{lemma:clinching_properties}. If $v_1 \leq B_1$, then
$\csw(x) = v_1 = \csw^*$, otherwise: $\csw(x) = B_1 \geq  \frac{1}{2} (B_1 +
v_2) \geq \frac{1}{2} \csw^*$.

In the remaining case, let $x,\pi$ be the outcome of the clinching auction and
let $[p_0, p_f]$ be the clinching interval. First, we give an upper bound to
$\csw^*$. Let
$\bar{x}^*$ be the outcome achieving $\csw^*$, then:
$$\csw^* = \sum_i \min(B_i, v_i \cdot \bar{x}_i^*) \leq \sum_{i; v_i
> p_0} B_i + \sum_{i; v_i \leq p_0} p_0 \bar{x}^*_i = \sum_{i; v_i > p_0}
B_i +  p_0 \left[ 1 - \sum_{i; v_i > p_0} \frac{B_i}{v_i} \right]^+$$
by the format of $x^*$ in Lemma \ref{lemma:x_star_format}. Now, we give two
lower bounds for $\csw(x)$ using Lemma
\ref{lemma:clinching_properties}. Let $p_f = v_k$. By our assumption that the
players have distinct values, there is only a single player with value $p_f$.
$$\csw(x) \geq \sum_{i} \pi_i = \sum_{i; v_i > p_f} B_i + \pi_k$$
Another way of bounding the revenue of the auction is to consider the ascending
price procedure and integrate the derivative over the supply over the price,
i.e.:
$$ \csw(x)  \geq \sum_{i} \pi_i = \int_{p_0}^{p_f} p \cdot
[-\partial_p S(p)] dp + \sum_{i; v_i \in [p_0, p_f]} v_i \cdot
[S(v_i-)-S(v_i)]$$
The first term represents the integral of the price over the derivative of the
remnant supply. Since the remnant supply is decreasing, we need to integrate
minus the value of the derivative. The sum in the second terms, takes care of
the discontinuities in the supply function between $v_i-$ (just before $v_i$)
and $v_i$. Now, we know by the definition of $\delta_j^i$ that
$S(v_i-)-S(v_i) = \sum_{j \in C(i)} \delta_j^i $, which gives us:
$$ \csw(x)  \geq p_0 \int_{p_0}^{p_f} [-\partial_p S(p)] dp + \sum_{i; v_i
\in [p_0, p_f)} v_i \cdot \frac{B_i}{v_i} \cdot \abs{C(v_i)} + v_k \cdot
\frac{B_k - \pi_k}{v_k} \cdot \abs{C(v_k)} $$
Now, we note that $\int_{p_0}^{p_f} [-\partial_p S(p)] dp +  \sum_{i; v_i
\in [p_0, p_f)}  \frac{B_i}{v_i} \cdot \abs{C(v_i)} +
\frac{B_k - \pi_k}{v_k} \cdot \abs{C(v_k)} = 1$ since it corresponds to the
total variation in the supply. So, we have a weighted sum of this total
variation, where the weights are all above $p_0$. From this observation
and the fact that $\abs{C(v_i)} \geq 1$ for $v_i \geq p_0$, we can
conclude that:
$$ \csw(x)  \geq \sum_{i; v_i \in [p_0, p_f)} B_i + (B_k - \pi_k) + p_0
\left[ 1-\sum_{i;v_i \in [p_0, p_f)} \frac{B_i}{v_i} - \frac{B_k -
\pi_k}{v_k}\right]^+$$
Combining this with the first bound we got on $\csw(x)$, we get:
$$2 \csw(x) \geq \sum_{i; v_i \geq p_0} B_i + p_0
\left[ 1-\sum_{i;v_i \in [p_0, p_f)} \frac{B_i}{v_i} - \frac{B_k -
\pi_k}{v_k}\right]^+ \geq \csw^*  $$
by the first bound obtained for $\csw^*$.
\end{proof}

A direct consequence from the proof is:

\begin{corollary}[revenue]
 If the clinching auction allocates items to more then one player, then its
revenue is at least $\frac{1}{2} \cdot \csw^*$.
\end{corollary}

The following example shows that our analysis is tight:

\begin{example}
 Consider two agents with $v_1 = 1$, $B_1 = \infty$ and $v_2 = \alpha \gg 1$,
$B_2 = 1$ and one unit of a divisible good. For such parameters the
allocation $\bar{x}^* = \left(1-\frac{1}{\alpha}, \frac{1}{\alpha}
\right)$ provides the optimal liquid welfare $\csw^* = 2-\frac{1}{\alpha}$.
The clinching auction generates allocation $x=(0,1)$ for which $\csw(x) = 1$.
As $\alpha \rightarrow \infty$, the ratio $\frac{\csw^*}{\csw(x)} \rightarrow
2$.
\end{example}

\section{A $2$-Approximation via Market Equilibrium}

We defined a quantifiable measure of efficiency (Section
\ref{sec:preliminaries}) and showed it can be approximated by an
incentive-compatible mechanism (Section \ref{sec:clinching}). The remaining
item in the list of desiderata was to show that our efficiency measure allows
for different designs. Here we show that we have ``an extra bunny in the hat'',
an auction that also achieves a $2$-approximation to the liquid welfare
objective and is \emph{not} based on Ausubel's clinching technique. Instead, it
is based on the concept of Market Equilibrium.

Borrowing inspiration from general equilibrium theory, consider a market with $n$ buyers each endowed
with $B_i$ dollars and willing to pay $v_i$ per unit for a certain divisible
good. This is the special case where there is only one product in the
market. In this case, a price $p$ is called a \emph{market clearing price} if
each buyer can be assigned an optimal basket of goods (in the particular of a
single product, an optimal amount of the good) such that there is no surplus or
deficiency of any good. Observe that there is one such price and
that allocations can be computed once the price is found. Our
Uniform Price Auction simply computes the market clearing price and allocates
according to it. This defines the allocation. The payments are computed using
the Myerson's formula for this allocation and happen to be different than the clearing price.

\begin{defn}[Uniform Price Auction]\label{defn:uniform_price_auction}
Consider $n$ agents with values $v_1 \geq \hdots \geq v_n$ (i.e., ordered without loss of generality) and budgets $B_i$.
Consider the auction that allocates one unit of a divisible good in the
following way: let $k$ be the maximum integer such that $\sum_{j=1}^k B_j \leq
v_k$, then:
\begin{itemize}
 \item Case I: if $\sum_{j=1}^k B_j > v_{k+1}$ allocate $x_i =
\frac{B_i}{\sum_{j=1}^k B_j}$ for $i=1,\ldots,k$ and nothing for the remaining players.
 \item Case II: if $\sum_{j=1}^k B_j \leq v_{k+1}$ allocate $x_i =
\frac{B_i}{v_{k+1}}$ for $i=1,\ldots,k$, $x_{k+1} = 1-\sum_{j=1}^k x_j$ and nothing for
the remaining players.
\end{itemize}
Payments are defined through Myerson's integral
(Lemma \ref{lemma:myerson}).
\end{defn}

Case I corresponds to the case where the market clearing price
of the Fisher Market instance is $p = \sum_{j=1}^k B_j$. Case II
coresponds to the case where the Market clearing price is $p = v_{k+1}$. First
we show that this auction induces an incentive-compatible auction that does not
exceed the budgets of the agents. (Proofs can be found in appendix
\ref{appendix:missing_proofs}.) Then we show that it is a $2$-approximation to
the liquid welfare benchmark.

\begin{lemma}[Monotonicity]\label{lemma:uniform_monotonicity}
 The allocation function of the Uniform Price Auction is monotone, i.e., $v_i
\mapsto x_i(v_i, v_{-i})$ is non-decreasing.
\end{lemma}

\begin{lemma}[Budget feasibility]\label{lemma:budget_feasibility}
The payments that make this auction incentive-compatible do not exceed the
budgets.
\end{lemma}

\begin{theorem}
 The Uniform Price Auction is an incentive compatible $2$-approximation to the
liquid welfare objective.
\end{theorem}

\begin{proof}
 Let $k$ be as in Definition \ref{defn:uniform_price_auction}. First we
establish an upper bound on $\csw^*$. If $\bar{x}^*$ is the allocation
achieving $\csw^*$, then:
$\csw^* = \sum_i \min(v_i \bar{x}_i^*, B_i) \leq \sum_{i=1}^k B_i +
\sum_{i=k+1}^n v_i \bar{x}^*_i \leq \sum_{i=1}^k B_i + v_{k+1}$.

Let $p$ be the market clearing price and $x$ the allocation of the Uniform
Price Auction. By the definition of $k$, $p \leq v_i$ for all $i \geq k$, so for
those players, $v_i x_i = v_i \frac{B_i}{p} \geq B_i$. Therefore $\csw(x) \geq
\sum_{i=1}^k B_i$. Now, we also show that $\csw(x) \geq v_{k+1}$. In Case I,
this is trivial since: $\csw(x) \geq \sum_{i=1}^k B_i > v_{k+1}$. In Case II,
$v_{k+1} \cdot x_{k+1} = v_{k+1} \cdot \left( 1-\sum_{j=1}^k \frac{B_j}{v_{k+1}}
\right) = v_{k+1} - \sum_{j=1}^k B_j < B_{k+1}$, so: $\csw(x) = \sum_{j=1}^k
B_k + v_{k+1} \cdot x_{k+1} = v_{k+1}$.

Summing up two inequalities, we have:
$\csw(x) \geq \frac{1}{2} \left[ \sum_{j=1}^k B_j + v_{k+1}\right]
\geq \frac{1}{2} \cdot \csw^*$.
\end{proof}

The same example used for showing that the analysis for the Clinching Auction
was tight can be used for showing that the analysis for the Uniform Price
Auction is tight.

\begin{example}
 Consider two agents with $v_1 = 1$, $B_1 = \infty$ and $v_2 = \alpha \gg 1$,
$B_2 = 1$ and one unit of a divisible good. We know that $\csw^* = 2 -
\frac{1}{\alpha}$. The market clearing price for this instance is $p = 1$,
which produces an allocation $x= (0,1)$ with $\csw(x) = 1$.
\end{example}

\subsection{Better than the Clinching Auction: An Instance by Instance Comparison}

One of the advantages in having a quantifiable measure of efficiency is that we
can compare two different outcomes and decide which one is ``better''. In this
section we show that although the worst-case guarantees of the clinching auction
and of the uniform-price auction are identical, the liquid welfare of the
uniform-price auction is \emph{always} (weakly) dominates that of the clinching
auction.

\begin{theorem}
Consider $n$ players with valuations $v_1 \geq \hdots \geq v_n$ and
budgets $B_1,\hdots, B_n$. Let $x^\cl$ and $x^\up$ be the outcomes of the
Clinching and Uniform Price Auctions respectively. Then: $\csw(x^\up) \geq
\csw(x^\cl)$.
\end{theorem}

\begin{proof}
 Let $k^\up$ be the value of $k$ as in Definition
\ref{defn:uniform_price_auction}. We can write the liquid welfare as
$\csw(x^\up) = \sum_{j=1}^{k^\up} B_j + v_{k^\up + 1} x^\up_{k^\up + 1}$, where
$v_{k^\up + 1} x^\up_{k^\up + 1} \leq B_{k^\up + 1}$. Also, for $i \leq k^\up$,
$B_i = x_i^\up \max\{v_{k^\up + 1}, \sum_{j=1}^{k^\up} B_j \} \geq x_i^\up
v_{k^\up + 1}$, therefore: $$\textstyle 1= \sum_i x^\up_i \leq
\sum_{i=1}^{k^\up} \frac{B_i}{v_{k^\up + 1}} + x^\up_{k^\up + 1} \qquad \qquad
(*)$$

 Let $\pi^\cl$ be the payments of the clinching auction for $v,B$. By Lemma
\ref{lemma:clinching_properties}, there exists $k^\cl$ such that for every $i \leq k^\cl$,
$\pi^\cl_i = B_i$ and for $i > k^\cl + 1$, $x^\cl_i = 0$. We can write the
liquid welfare as: $\csw(x^\cl) = \sum_{j=1}^{k^\cl} B_j + \min\{B_{k^\cl
+1}, v_{k^\cl +1} \cdot x^\cl_{k^\cl +1} \}$. Also, the final price (as in Definition
\ref{defn:clinching_interval}) is $v_{k^\cl + 1}$, and therefore $B_i =
\pi^\cl_i \leq v_{k^\cl + 1} x_i^\cl$. Therefore: $$\textstyle 1 = \sum_i
x^\cl_i \geq \sum_{i=1}^{k^\cl} \frac{B_i}{v_{k^\cl + 1}} + x^\cl_{k^\cl + 1}
\qquad \qquad (**)$$

First we argue that $k^\cl \leq k^\up$. Assume for contradiction that $k^\cl >
k^\up$, then:
$$\textstyle v_{k^\up + 1} \stackrel{*}{\leq} \sum_{i=1}^{k^\up} B_i +
v_{k^\up + 1} x^\up_{k^\up + 1} \leq \sum_{i=1}^{k^\up + 1} B_i \leq
\sum_{i=1}^{k^\cl } B_i \stackrel{**}{\leq} v_{k^\cl + 1} (1-x^\cl_{k^\cl +
1}) \leq v_{k^\cl + 1} \leq v_{k^\up + 1} $$
where the last inequality comes from $k^\cl > k^\up$. This would imply that all
inequalities above hold with equalities, in particular: $B_{k^\up + 1} =
v_{k^\up + 1} x^\up_{k^\up + 1}$, and therefore
$x^\up_{k^\up + 1}=\frac {B_{k^\up + 1}} {v_{k^\up + 1}}$. Recall that we also
have that $x^\up_{k^\up + 1} = 1 -
\sum_{i=1}^{k^\up} \frac{B_i}{v_{k^\up + 1}}$, and therefore $v_{k^\up
+ 1} = \sum_{j=1}^{k^\up + 1} B_j$, contradicting the definition of $k^\up$.

So, we proved in the previous paragraph that  $k^\cl \leq k^\up$. If $k^\cl <
k^\up$, then $\csw(x^\cl) \leq \sum_{i=1}^{k^\cl + 1} B_i \leq
\sum_{i=1}^{k^\up} B_i \leq \csw(x^\up)$. Now, if
$k^\cl = k^\up$, then $(*)$ and $(**)$ together imply that $x^\cl_{k^\cl+1} \leq
x^\up_{k^\up+1}$, hence $\csw(x^\cl) = \sum_{i=1}^{k^\cl} B_i + \min\{B_{k^\cl
+ 1}, v_{k^\cl + 1} \cdot x^\cl_{k^\cl+1} \} \leq \sum_{i=1}^{k^\up} B_i +
\min\{B_{k^\up + 1}, v_{k^\up + 1} \cdot x^\up_{k^\up+1} \} = \csw(x^\up)$.
\end{proof}

\section{A Lower Bound and Some Matching Upper Bounds}

In the previous sections, we showed two different auctions that are incentive
compatible $2$-approximations to the optimal liquid welfare for the setting
of multi-unit auctions with additive valuations. In this section we investigate
the limits of the approximability of the liquid welfare. By the observation
depicted in Figure \ref{fig:xstar}, it is clear that an exact incentive
compatible mechanism is not possible for this setting. First, we present a
$\frac{4}{3}$ lower bound and show matching upper bounds for some special cases.
%

\begin{theorem}\label{thm:lower-bound}
For the  multi-unit setting with additive
values and public budgets, there is no incentive-compatible mechanism that
approximates the liquid welfare objective by a factor better then
$\frac{4}{3}$.
\end{theorem}

\begin{proof}
Consider the problem of selling one divisible item to two agents that have
equal budgets $B_1 = B_2 = 1$. Let $x(v_1, v_2)$ be the allocation rule for an
auction that is a $\gamma$-approximation to $\csw^*$. Also, let $\bar{x}^*(v_1,
v_2)$ be solution maximizing the $\csw$.

Fix some number $\alpha > 1$ and consider the allocations of the auction for valuations $(1, \alpha)$, $(\alpha,
1)$, $(\alpha, \alpha)$. Since the mechanism is truthful, the allocation should
be monotone (Lemma \ref{lemma:myerson}) so: $x_1(1,\alpha) \leq x_1(\alpha,
\alpha)$ and $x_2(\alpha,1) \leq x_2(\alpha, \alpha)$. Since $x_1(\alpha,
\alpha) + x_2(\alpha, \alpha) \leq 1$ we get that one of the summands is at most $\frac 1 2$. Let us assume that $x_1(\alpha,
\alpha) \leq \frac 1 2$ and thus $x_1(1,
\alpha) \leq \frac 1 2$ (the other case is similar). Since $\csw(\bar{x}^*(1,\alpha)) =
\csw(\bar{x}^*(\alpha,1)) = 2-\frac{1}{\alpha}$, the $\gamma$ approximation
implies that:
$$\gamma^{-1} ( 2- \textstyle\frac{1}{\alpha}) \leq \min(1, x_1(1,\alpha)) +
\min(1, \alpha \cdot x_2(1,\alpha)) \leq x_1(1,\alpha) + 1 \leq 1.5$$
As $\alpha$ approaches $\infty$ we get that $\gamma$ approaches $\frac{4}{3}$.
\end{proof}

\subsection{A Matching Upper Bound for $2$ Bidders with Equal Budgets}

For the special case of $2$ players and equal budgets, we give a matching upper
bound.  Up to re-scaling values and budgets, we can assume that the
players have all budgets equal to $1$. For this special case, consider the
following auction.

\begin{defn}[$\frac{4}{3}$-approx for $\csw^*$]\label{defn:34auction}
Consider the following auction for $1$ divisible good and two bidders with
(known) budgets $B_1 = B_2 = 1$. It maps values $(v_1, v_2)$ to an allocation
$x = (x_1,x_2)$ and is symmetric (i.e., for all $v$ we have that $x_1(v,v)=x_2(v,v)$). So,
we only need to specify $x(v_1, v_2)$ for $v_1 \geq v_2$:
\begin{itemize}
 \item if $v_1 = v_2$, $x(v_1,v_2) = (\frac{1}{2}, \frac{1}{2})$
 \item if $v_2 \leq \frac{1}{3}$, $x(v_1, v_2) = (1,0)$
 \item if $\frac{1}{3} \leq v_2 \leq 1$, $x(v_1, v_2) =
(\frac{1}{4}+\frac{1}{4 v_2}, \frac{3}{4}-\frac{1}{4 v_2})$.
 \item if $1 \leq v_2 $, $x(v_1, v_2) = (\frac{1}{2}, \frac{1}{2})$
\end{itemize}
Payments are calculated using Myerson's Lemma \ref{lemma:myerson}.
\end{defn}

One can verify that the allocation rule is monotone. Moreover, the payments that
make this auction truthful do not exceed the
budget, since for any $v_i > 1$, $x_i(v_i, v_{-i}) = x_i(1, v_{-i})$. It remains
to prove the approximation guarantee, i.e.,  that $\csw(x) \geq \frac{3}{4}
\csw^*$.

\begin{lemma}\label{lemma:approx_34_auction}
 The approximation ratio of the auction in Definition \ref{defn:34auction} is
$4/3$, i.e., $\csw(x) \geq \frac{3}{4} \csw^*$.
\end{lemma}

The proof is by case analysis and can be found in appendix
\ref{appendix:missing_proofs}.

\section{An $O(\log^2 n)$-approximation for Subadditive Players with
Private Budgets}

In this section we consider the setting where players have subadditive
valuations and private budgets. This is a notoriously hard setting for
Pareto-optimality. In fact, considering either subadditive valuations or
privated budgets alone already produces an impossibility result for achieving
Pareto-efficient outcomes.

We will have one divisible good and each player has
a subadditive valuation $v_i:[0,1] \rightarrow \R_+$ and a budget $B_i$. This
setting differs from the previously considered in the sense that budgets $B_i$
are private information of the players.

The auction we propose is inspired in a technique by Bartal, Gonen and Nisan
\cite{BGN}. To describe it, we use the following notation: $\bv(x_i) = \min\{v_i(x_i), B_i\}$. Now,
consider the following selling procedure:

\begin{defn}[Sell-Without-$r$]
 Let $r$ be a player. Consider the following mechanism to sell
half the good, to players $i \neq r$ using the information
about $\bv_r(\frac{1}{2})$.

Divide the segment $[0,\frac{1}{2}]$ into $k = 8 \log(n)$ parts, each of size
$\frac{1}{2k}$. Associate part $i = 1,\ldots,k$ with price per unit $p_i =
\frac{2^i}{8} \bv_r(\frac{1}{2})$. Order arbitrarily all players but player $r$. Each player different than $r$, in his turn,
takes his most profitable (unallocated) subset of $[0,\frac 1 2]$ under the specified prices. Players are not allowed to pay more than their budget.

More precisely, let $p:[0,\frac{1}{2}] \rightarrow \R_+$ be such that for $x \in
[\frac{1}{2k}(i-1),\frac{1}{2k}i]$, $p(x) = p_i = \frac{2^i}{8}
\bv_r(\frac{1}{2})$. Now, for $i=1,\ldots,r-1,r+1,\ldots, n$, let $x_i$ maximize
$v_i(x_i) - \int_{z_i}^{z_i+x_i} p(t) dt$ where $z_i = \sum_{j < i} x_j$,
conditioned on the payment being below the budget, i.e., $\int_{z_i}^{z_i+x_i}
p(t) dt \leq B_i$. Set the payment as: $\pi_i = \int_{z_i}^{z_i+x_i} p(t) dt$.
\end{defn}

The subroutine Sell-Without-$r$ is used in our main construction for this section:

\begin{defn}[Estimate-and-Price]
 Given one divisible good and $n$ players with valuations $v_i(\cdot)$ and
budgets $B_i$, consider the following auction: let $r_1 =
\arg\max_i \bv_i(\frac{1}{2})$ and $r_2 = \arg\max_{i \neq r_1}
\bv_i(\frac{1}{2})$. We say that $r_1$ is the \emph{pivot player}. Let $(x,\pi)$ be the outcome of Sell-Without-$r_1$ for
players $[n]\setminus r_1$ and let $(x',\pi')$ be the outcome of
Sell-Without-$r_2$ for players $[n]\setminus r_2$.

For players $i \neq r_1$, allocate $x_i$ and charge $\pi_i$. For $r_1$
if $v_{r_1}(x'_{r_1})-\pi'_{r_1} \geq v_{r_1}(\frac{1}{2}) -
2 \cdot \bv_{r_2}(\frac{1}{2})$ allocate him $x'_{r_1}$ and charge $\pi'_{r_1}$
and if not, allocate $\frac{1}{2}$ and charge $2 \cdot \bv_{r_2}(\frac{1}{2})$.
\end{defn}

First, notice that the auction defined above is feasible, since $r_1$ is
allocated at most half of the good and the players in $[n] \setminus r_1$ get
allocated at most half of the good. Now we show that this auction is incentive
compatible:

\begin{lemma}\label{lemma:truthfulness-1}
 The Estimate-and-Price auction is incentive compatible for players with
private budgets.
\end{lemma}

\begin{proof}
We argue that no deviation in which a player changes his value and his budget
can be profitable. For a player $i \neq r_1$, if $i$ deviates and does not
become the pivot player, his allocation and prices remain the same. Now, this player could try to deviate to
become the pivot player. However, in this case he either gets allocated as before, or
he is allocated $\frac{1}{2}$ and pays $2 \bv_{r_1}(\frac{1}{2}) >
\bv_i(\frac{1}{2})$, either exceeding his budget or getting negative utility.

As for the pivot player $r_1$, he cannot benefit from decreasing $\bv_{r_1}(\frac{1}{2})$ so that he will
no longer be the pivot player (all other changes do not change his allocation and payment), since $r_2$ will become the pivot player in this case and $r_1$ will be allocated as in the
Sell-Without-$r_2$. This is uselss for $r_1$: if $r_1$ was allocated half of the good when playing truthfully, this means
that the profit of $r_1$ could only decrease.
If $r_1$ was allocated as in Sell-Without-$r_2$ when playing truthfully, then misreporting did not change his allocation and payment.
\end{proof}

Notice that the price offered to the pivot player is $2 \cdot \bv_{r_2}(\frac{1}{2})$ and not $\bv_{r_2}(\frac{1}{2})$ as the reader who is familiar with \cite{BGN}
would expect. The change is because an auction on one item where every player bids the minimum of his budget and value and the winner pays
the second highest bid is not truthful. To see that, consider two players with identical budgets that their value for the item exceed their budgets.
Each of the players has an incentive to report a slight increase in the budget to win the item.

\subsection*{An $O(\log n)$-Approximation for Submodular Bidders}

We now
show that the Estimate-and-Price auction is a poly-log approximation for
subadditive bidders. For clarity of exposition, we first show that for the
special case of \emph{submodular bidders}, the Estimate-and-Price auction is a
$O(\log n)$ approximation. Then we show how to modify the proof to handle
the wider class of subadditive valuations losing only an $O(\log n)$ factor.

Recall that a valuation function $v_i$ is subadditive if
$v_i:[0,1] \rightarrow \R_+$ such that $v_i(x_i+y_i) \leq v_i(x_i) + v_i(y_i)$.
The special case of submodular bidders is characterized by concave $v_i$
functions.

Now, we prove two lemmas towards the proof that the Estimate-and-Price
auction is a $O(\log n)$-approximation. First, we claim that not all items are
sold in the Sell-Without-$r_1$ procedure:

\begin{lemma}\label{lemma:not-all-sold}
 Let $x$ be the outcome of the Sell-Without-$r_1$ subroutine of the
Estimate-and-Price Auction, then $\sum_{i \neq r_1} x_i <
\frac{1}{2}$.
\end{lemma}

\begin{proof}
 We note that $\sum_{i \neq r_1} \pi_i \leq \sum_{i \neq r_1} \bv_i(x_i) \leq
\sum_{i \neq r_1} \bv_i(\frac{1}{2}) \leq n \cdot \bv_{r_1}(\frac{1}{2})$ by
the definition of $r_1$. Now, if we were to sell $\frac{1}{2}$ to $[n]
\setminus r_1$, we would have $\sum_{i \neq r_1} \pi_i = \int_0^{1/2} p(t) dt
\geq \int_{k-1/2k}^{1/2} p(t) dt = \frac{1}{2k} \cdot \frac{2^k}{8}
\bv_{r_1}(\frac{1}{2}) = \frac{n^8}{16 \cdot 8 \cdot \log(n)} \cdot
\bv_{r_1}(\frac{1}{2}) > n \cdot \bv_{r_1}(\frac{1}{2})$ since for $n \geq 1.8$,
$n^7 > 16 \cdot 8 \cdot \log(n)$.
\end{proof}

\begin{lemma}\label{lemma:dagger_defn}
 Let $x^\dagger$ be the solution of $\max \csw(x^\dagger)$ s.t. $x^\dagger_{r_1}
= 0$ and $\sum_{i \neq r_1} x^\dagger_i = \frac{1}{2}$, then
$\bv_{r_1}(\frac{1}{2}) + \csw(x^\dagger) \geq \frac{1}{2} \csw^*$.
\end{lemma}

\begin{proof}
 Let $x^*$ be the outcome s.t. $\csw^* = \csw(x^*)$, then
$\bv_{r_1}(\frac{1}{2}) + \csw(x^\dagger) \geq \csw(\frac{1}{2} x^*) \geq
\frac{1}{2}\csw(x^*)$, where
the first inequality comes from the fact $\frac{1}{2} \geq \frac{1}{2}
x^*_{r_1}$ and that
$\frac{1}{2}x^*_{-r_1}$ satisfies the requirements of the program defining
$x^\dagger$. The second inequality comes
from the concavity of $\csw$.
\end{proof}

\begin{lemma}\label{lemma:price_bar_p}
For submodular bidders, let $x$ be the outcome of the Estimate-and-Price auction
and $\bar{p}$ be the price of the cheapest unsold item, i.e., $\bar{p} = \lim_{t
\downarrow (\sum_{i \neq r_1} x_i)} p(t)$. Let also $x'$ be any allocation. Then
for each $i \neq r_1$, $\bv(x_i) \geq \bv(x'_i) - \bar{p} \cdot x'_i$.
\end{lemma}

\begin{proof}
 If $x_i \geq x'_i$, this is true by monotonicity of $\bv_i$. Now, if not, then
player $i$ did not acquire more goods, then it must have been for one of two
reasons: either he exhausted his budget or the price exceeds his marginal
value. Notice Lemma
\ref{lemma:not-all-sold} shows that there are items that are available.

In the first case: $\pi_i = B_i$, so $\bv_i(x_i) = B_i \geq \bv_i(x'_i)$. In
the second case, $v_i(x_i) \geq v_i(x_i + \epsilon) -  p' \cdot \epsilon$ for
all $\epsilon \in [0,\epsilon_0]$ for some small $\epsilon_0$. By concavity of
$v_i(\cdot)$, $v_i(x'_i) - v_i(x_i) \leq p' \cdot (x'_i - x_i)\leq \bar{p}
\cdot x'_i$.  So, $\bv_i(x_i) = v_i(x_i) \geq v_i(x'_i) - \bar{p}
\cdot x'_i \geq \bv_i(x'_i) - \bar{p} \cdot x'_i$.
\end{proof}

Now we are ready to prove our main result:

\begin{theorem}
For submodular bidders, the Estimate-and-Price auction is an $O(\log
n)$-approximation to the
liquid welfare objective.
\end{theorem}

\begin{proof}
 We analyze two different cases, depending on which fraction of the optimal
liquid welfare is produced by player $r_1$. Let $x^\dagger$ be as in the
statement of Lemma \ref{lemma:dagger_defn}. Then we consider:\\

\noindent \emph{Case I}: $\bv_{r_1}(\frac{1}{2}) \geq 4 \csw(x^\dagger)$.

We note that the utility of $r_1$ is at least
$v_{r_1}(\frac{1}{2})-2 \cdot \bv_{r_2}(\frac{1}{2})$ and that
$\bv_{r_2}(\frac{1}{2})
\leq \csw(x^\dagger)$. Therefore: $u_{r_1} \geq v_{r_1}(\frac{1}{2}) - 2 \cdot
\csw(x^\dagger) \geq \frac{1}{2} \bv_{r_1}(\frac{1}{2}) \geq 2 \csw(x^\dagger)$ which
implies that $\bv_{r_1}(\frac{1}{2}) + \csw(x^\dagger) \leq 2 u_{r_1} +
\frac{1}{2} u_{r_1} = \frac{5}{2} u_{r_1}$.
Using Lemma \ref{lemma:dagger_defn}, we get that: $\csw(x) \geq
\min\{B_{r_1, }u_{r_1}\} \geq
\frac{2}{5}[ \bv_{r_1}(\frac{1}{2}) + \csw(x^\dagger) ] \geq \frac{2}{5}
\csw^*$.\\

\noindent\emph{Case II}: $\bv_{r_1}(\frac{1}{2}) < 4 \csw(x^\dagger)$.

Consider the price $\bar{p}$ as defined in Lemma \ref{lemma:price_bar_p}.
If $\bar{p} = p_1 = \frac{1}{4} \bv_{r_1}(\frac{1}{2})$, then $\csw(x) \geq
\csw(x^\dagger) - p_1 \sum_i x^\dagger_i = \csw(x^\dagger) -
\frac{1}{8} \bv_{r_1}(\frac{1}{2}) > \frac{1}{2}\csw(x^\dagger)$. Therefore:
$\textstyle \csw^* \leq 2 \cdot [ \bv_{r_1}(\frac{1}{2}) +
\csw(x^\dagger) ] < 10 \cdot \csw(x^\dagger) < 20 \cdot \csw(x)$.
If on the other hand $\bar{p} = p_k > p_1$, then:
$$\csw(x) \geq \sum_{i \neq r_1} \pi_i \geq \sum_{i=1}^{k-1} p_i
\cdot \frac{1}{2k} \geq (\bar{p} - p_1) \cdot \frac{1}{2k} \geq \frac{1}{2}
\cdot \bar{p} \cdot \frac{1}{2k} \geq \frac{\bar{p}}{32 \cdot \log(n)}$$
where the two last inequalities come from the fact that the price doubles every
interval. Now, we know that: $16 \log(n) \cdot \csw(x) \geq \frac{1}{2}
\bar{p}$. Using that together with Lemma \ref{lemma:price_bar_p} for $x'_i =
x^\dagger_i$, we get: $(16 \log(n) + 1) \csw(x) \geq \csw(x^\dagger)$. Since in
this case: $\csw(x^\dagger) \geq \frac{1}{10} \csw^*$ we have that $\csw^* \leq
O(\log n) \cdot \csw(x)$.
\end{proof}

\subsection*{An $O(\log^2 n)$-Approximation for Subadditive Bidders}
The only point we used in the previous proof that $v_i$ is concave is in Lemma
\ref{lemma:price_bar_p}. For proving every other argument the
subadditivity of $v_i$ suffices. To prove the result for subadditive valuations, we
re-define $x^\dagger$ as the argmax of $\csw(x^\dagger)$ conditioned on
$x^\dagger_i = 0$, $\sum_i x_i^\dagger = \frac{1}{2}$ and $x_i^\dagger \leq
\frac{1}{2k}$. The extra condition $x_i^\dagger \leq \frac{1}{2k}$ implies that
in the proof of Lemma \ref{lemma:price_bar_p} no player ever would need to pay
more the $2 \bar{p}$ per unit. So, we can show that $\bv_i(x_i) \geq
\bv_i(x_i^\dagger) - 2 \bar{p} \cdot x_i^\dagger$.

The only difference is that under this new definition that gap between $\csw^*$
and $\bv_{r_1}(\frac{1}{2}) + \csw^*$ is not constant anymore, but $O(k) =
O(\log n)$. This difference makes us lose another $O(\log n)$ factor in the approximation ratio:

\begin{theorem}
For subadditive bidders, the Estimate-and-Price auction is an
$O(\log^2 n)$-approximation to the liquid welfare objective.
\end{theorem}

\subsection*{Extension to Indivible Goods}
We note that all techniques in
this section generalize to indivisible goods. For $m$ identical indivisible
goods, one can get the an $O(\log m)$ approximation by dividing the items in
$O(\log m)$ groups and setting prices $p_i = 2^i \cdot
O(\bv_{r_1}(\frac{1}{2}))$, where $v_{r_1}(\frac{1}{2})$ now stands for the
value of player $r_1$ for half the good.

\bibliography{sigproc}
%
%
\appendix
\section{Missing Proofs}\label{appendix:missing_proofs}

\begin{proofof}{Lemma \ref{lemma:uniform_monotonicity}}
  Fix a valuation profile $v$ with $v_1 \geq \hdots \geq v_n$ and let $k$ be as
in Definition \ref{defn:uniform_price_auction}. The agents
$i=1,\ldots,k$ do not change their allocation if they increase their value. For
a player $i > k+1$, their allocation can go from zero to non-zero once their
value is so high that they become the $k+1$ player. It remains to consider the
$k+1$ player: as he increases his allocation and continues to have the $k+1$
highest value, his allocation increases, since it is 
$x_{k+1} = 1-\sum_{j=1}^k \frac{B_j}{v_{k+1}}$.

Now, two things can happen while $v_{k+1}$ increases:
\begin{itemize}
 \item the value of $v_{k+1}$ reaches $\sum_{j=1}^{k+1} B_j$ and the allocation
gets updated to $ \frac{B_{k+1}}{\sum_{j=1}^{k+1} B_j}$. At this point, the
allocation of $x_{k+1}$ continues the same as $v_{k+1}$ increases.

 \item the value of $v_{k+1}$ reaches $v_k$ and displaces $k$ and the $k$-th
highest value. One of two things happen: (i) if $\sum_{j=1}^{k-1} B_j + B_{k+1}
> v_k$, then the market clearing price becomes $v_k = v_{k+1}$ and therefore
the allocation $x_{k+1}$ gets updated to $1-\sum_{j=1}^{k-1} \frac{B_j}{v_{k}}
\geq 1-\sum_{j=1}^{k} \frac{B_j}{v_{k}}$; or (ii) if $\sum_{j=1}^{k-1} B_j +
B_{k+1} \leq v_k$, now the market clearing price becomes $v_{k}$ and $k+1$ is
now allocated as $\frac{B_{k+1}}{v_k} \geq 1-\sum_{j=1}^{k} \frac{B_j}{v_{k}} $.
\end{itemize}
\end{proofof}

\begin{proofof}{Lemma \ref{lemma:budget_feasibility}}
For players $i > k+1$, this is trivial, since they do not get goods and pay
zero. For the rest of the players we look at the two cases:
\begin{itemize}
 \item Case I : player $k+1$ also does not get goods and pays zero. For the rest
of the player, their allocation is constant for any value $v'_i \geq
\sum_{j=1}^k B_j$, so, their payment is bounded by $(\sum_{j=1}^k B_j) \cdot x_i
= B_i$.
 \item Case II: player $k+1$ pays at most $v_{k+1} x_{k+1} = v_{k+1} \cdot
\left( 1- \sum_{j=1}^k \frac{B_j}{v_{k+1}}\right) < B_{k+1}$, since
$v_{k+1} < \sum_{j=1}^{k+1} B_j$ by the
definition of $k$. For the rest of the players $i \leq k$, their allocation is
constant for all $v'_i \geq v_{k+1}$, so their payment is bounded by $v_{k+1}
\cdot x_i = B_i$.
\end{itemize}
\end{proofof}

\begin{proofof}{Lemma \ref{lemma:approx_34_auction}}
 The proof is by case analysis. We assume $v_1 \geq v_2$, since $v_2 \geq v_1$
is analogous.
\begin{itemize}
 \item Case I : $v_1 \geq v_2$ and $v_2 \leq \frac{1}{3}$

\begin{itemize}
\item I.1 : if $v_1 \leq 1$, then $\csw(x) = v_1 = \csw^*$.
\item I.2 : if $v_1 > 1$, then: $\csw^* = 1
+ \min(1, v_2(1-\frac{1}{v_1})) \leq 1 + v_2 \leq \frac{4}{3}$ and $\csw(x) =
1$.
\end{itemize}

 \item Case II: $v_1 \geq v_2 \geq 1$

\begin{itemize}
\item II.1: if  $v_1 + v_2 \geq 3$, then: $\csw^* \leq 2$ and $\csw(x) =
\min(1,\frac{1}{2}v_1) + \min(1,
\frac{1}{2} v_2) \geq \frac{3}{2}$.
\item II.2: if $v_1 + v_2 \leq 3$ then $\csw^* = 1 +
\min(1, v_2(1-\frac{1}{v_1})) = 1 + v_2(1-\frac{1}{v_1})$ since for $1\leq
v_1,v_2 \leq 2$, it is easy to see that $v_2(1-\frac{1}{v_1})
\leq 1$. Also: $\csw(x) = \frac{1}{2}(v_1 + v_2)$.

We want to show that $\frac{1}{2}(v_1 + v_2) \geq \frac{3}{4} \cdot \left[
1+ v_2(1-\frac{1}{v_1}) \right]$. Re-writting that, we have that this is
equivalent to: $v_2 \geq \frac{2 v_1 - 3}{1 - 3/v_1}$. Taking derivatives, we
see that $\frac{2 v_1 - 3}{1 - 3/v_1}$ is monotone decreasing in the interval
$[1,2]$, so for this interval: $\frac{2 v_1 - 3}{1 - 3/v_1} \leq \frac{1}{2}
\leq 1 \leq v_2$.
\end{itemize}

\item Case III: $v_1 \geq v_2$, $\frac{1}{3} \leq v_2 \leq 1$

\begin{itemize}
\item III.1: $v_1 \leq 1$, then $\csw^* = v_1$ and $\csw(x) =
v_1 \cdot (\frac{1}{4} + \frac{1}{4v_2}) + v_2 \cdot (\frac{3}{4} -
\frac{1}{4v_2})$. Therefore, what we want to show is that: $v_1 \cdot
(\frac{1}{4} + \frac{1}{4v_2}) + v_2 \cdot (\frac{3}{4} -
\frac{1}{4v_2}) \geq \frac{3}{4} v_1$. We can re-write this as $3v_2 -1 \geq
(2-\frac{1}{v_2}) v_1$. For $\frac{1}{3} \leq v_2 \leq \frac{1}{2}$, this is
trivially true by sign analysis, since the left hand side is non-positive and
the right hand side is non-negative. For, for $\frac{1}{2} \leq v_2 \leq 1$,
this is equivalent to $v_1 \leq \frac{3 v_2 -1}{2-1/v_2}$, we note that the
function $\frac{3 v_2 -1}{2-1/v_2}$ is monotone non-increasing in the
interval $[\frac{1}{2},1]$ so: $\frac{3 v_2 -1}{2-1/v_2} \geq \frac{3}{2} \geq
v_1$.
\item III.2: $v_1 \geq 1$ and  $1 \leq v_1 \cdot (\frac{1}{4} + \frac{1}{4
v_2})$, then: $\csw^* = 1 +
\min(1, v_2(1-\frac{1}{v_1})) = 1 + v_2(1-\frac{1}{v_1})$ and $\csw(x) =
\min(1, v_1 \cdot (\frac{1}{4} + \frac{1}{4 v_2})) + v_2 \cdot (\frac{3}{4} -
\frac{1}{4 v_2}) = 1 + v_2 \cdot (\frac{3}{4} - \frac{1}{4 v_2})$. Therefore:
$\csw(x) = \frac{3}{4}(1+v_2) \geq \frac{3}{4} \csw^*$.

\item III.3: $v_1 \geq 1$ and  $1 > v_1 \cdot (\frac{1}{4} + \frac{1}{4
v_2})$, then: $\csw^* =  = 1 + v_2(1-\frac{1}{v_1})$ and $\csw(x) =
\min(1, v_1 \cdot (\frac{1}{4} + \frac{1}{4 v_2})) + v_2 \cdot (\frac{3}{4} -
\frac{1}{4 v_2}) = v_1 \cdot (\frac{1}{4} + \frac{1}{4 v_2}) + v_2 \cdot
(\frac{3}{4} - \frac{1}{4 v_2})$. We want to show that $v_1 \cdot (\frac{1}{4} +
\frac{1}{4 v_2}) + v_2 \cdot
(\frac{3}{4} - \frac{1}{4 v_2}) \geq \frac{3}{4} (1 + v_2(1-\frac{1}{v_1})) $,
which can be re-written as $\frac{v_1}{4} + \frac{v_1}{v_2} \geq 1$, which
follows from $v_1 \geq v_2$.
\end{itemize} \end{itemize} 
\end{proofof}


\end{document}